\documentclass[prl,aps,twocolumn, reprint, amsmath,amssymb,showpacs,superscriptaddress]{revtex4}

\usepackage{graphicx}% Include figure files
\usepackage{dcolumn}% Align table columns on decimal point
\usepackage{bm}
\usepackage{graphicx}
\usepackage{epsfig}
\usepackage{epsf}
\usepackage{amssymb}
\usepackage{amsmath}
\usepackage{amsthm}
\usepackage{multirow}
\usepackage{cases}
\usepackage[colorlinks=true,linkcolor=blue,citecolor=blue,pdfauthor={ },pdftitle={ },pdfsubject={ },pdfkeywords={ }]{hyperref}
\usepackage{tikz}
\usetikzlibrary{mindmap,trees, backgrounds}

\newtheorem*{theorem}{Theorem}

\newtheorem{proposition}{Proposition}

\newtheorem{corollary}{Corollary}

\usepackage{algorithm}
\usepackage{algpseudocode}
\usepackage{algpascal}
\alglanguage{pseudocode}

\begin{document}

\title{Approximation of Reachable Set for Coherently Controlled Open Quantum Systems: Application to Quantum State Engineering}

\author{Jun Li}
\affiliation{Hefei National Laboratory for Physical Sciences at Microscale and Department of Modern Physics, University of Science and Technology of China, Hefei, Anhui 230026, China}

\author{Dawei Lu}
\affiliation{Institute for Quantum Computing and Department of Physics and Astronomy, University of Waterloo, Waterloo, Ontario N2L 3G1, Canada}

\author{Zhihuang Luo}
\affiliation{Hefei National Laboratory for Physical Sciences at Microscale and Department of Modern Physics, University of Science and Technology of China, Hefei, Anhui 230026, China}

\author{Raymond Laflamme}
\affiliation{Institute for Quantum Computing and Department of Physics and Astronomy, University of Waterloo, Waterloo, Ontario N2L 3G1, Canada}
\affiliation{Perimeter Institute for Theoretical Physics, Waterloo, Ontario N2L 2Y5, Canada}
\affiliation{Canadian Institute for Advanced Research, Toronto, Ontario M5G 1Z8, Canada}

\author{Xinhua Peng}
\email{xhpeng@ustc.edu.cn}
\affiliation{Hefei National Laboratory for Physical Sciences at Microscale and Department of Modern Physics, University of Science and Technology of China, Hefei, Anhui 230026, China}
\affiliation{Synergetic Innovation Center of Quantum Information $\&$ Quantum Physics,
University of Science and Technology of China, Hefei, Anhui 230026, China}

\author{Jiangfeng Du}
\email{djf@ustc.edu.cn}
\affiliation{Hefei National Laboratory for Physical Sciences at Microscale and Department of Modern Physics, University of Science and Technology of China, Hefei, Anhui 230026,  China}
\affiliation{Synergetic Innovation Center of Quantum Information $\&$ Quantum Physics,
University of Science and Technology of China, Hefei, Anhui 230026, China}

\begin{abstract}
Precisely characterizing and controlling realistic open quantum systems is one of the most challenging
and exciting frontiers in quantum sciences and technologies.
In this Letter, we present methods of approximately computing reachable sets for coherently controlled dissipative systems, which is very useful for assessing control performances. We apply this to a two-qubit nuclear magnetic resonance spin system and implement some tasks of quantum control in open systems at a near optimal performance in view of purity: e.g., increasing polarization and preparing pseudo-pure states. Our work shows interesting and promising applications of environment-assisted quantum dynamics.
\end{abstract}


\maketitle

Recent years have seen immense advances in active and precise manipulation of a broad variety of quantum systems. The subject of quantum system control has been developed into a rapidly growing area \cite{DP} attracting substantial interests from the community of quantum information physicists. One of the fundamental tasks is to design reliable control techniques for systems that are exposed to a dissipative environment \cite{KKSS}.
As dissipation tends to irreversibly affect the system dynamics, it is recognized as one dominant source for information loss and hence must be suppressed.
Only recently was it realized that open system engineering may exhibit surprising advantages in some important aspects \cite{RMBL, Entanglement, SW}. For example, it was shown that the purification efficiency of heat-bath algorithmic cooling protocol can surpass the closed system limit \cite{RMBL}. In other researches \cite{Entanglement}, there emerged great interests in environment-assisted entangled  state engineering.
Dissipative production of entangled steady states has already been realized
in various experimental setups like trapped ions \cite{Lin}, superconducting circuit \cite{RTJS} and double quantum dot \cite{SKVCG}.

Although some ideas borrowed from classical control theory (e.g., time optimal control) have been successfully extended to construct methods for steering closed quantum systems \cite{D08}, it turns out to be more challenging  for open quantum systems.
The major reason  comes from the fact proved in \cite{A} that for a finite dimensional Markovian quantum system, coherent means of control cannot fully compensate the irreversibility of the dynamics. In fact, to what extent can the system evolving tendency be changed depends upon not only the external operations but also the structure of the relaxation mechanisms. This certainly  forms an obstacle in devising of general control methodology.
Previous research results have been able to characterize the reachable set on the states of a single qubit both qualitatively \cite{A} and quantitatively \cite{Y, RBR}. However, to generalize these results to higher dimensional systems is not easy \cite{R}.

Realising the lack of exact theory for the reachability problem, we propose to use approximation techniques. Basically the idea is to approximate the reachable set, usually in terms of simple geometric objects  (e.g., convex polytopes \cite{M14}, ellipsoids \cite{KV}), from both externally (over-approximation) and internally (under-approximation) \cite{ABDM}. Although various strategies have been put forward, computing reachable set in general remains a challenging task \cite{M14}. This is true especially for nonlinear systems, and quantum control models are indeed recognized as nonlinear \cite{PY}.
%Purity as an important concept quantifying the incoherent impacts from the environment, is particularly suited for studying how relaxation noises impose restricts on the achievable region of states. For example, one basic result for \emph{unital} systems (where the equilibrium state is the maximally mixed state) states that the purity function must be monotonically decreasing with time regardless of the controls \cite{LSA}. For the case of non-unital dynamics the situation is more complicated since purification can occur, depending on which kind of control protocol is applied. It is thus natural to consider that, given a practical relaxation process and a realistic control protocol, what is the upper bound of purity that the system can not surpass.
In this Letter, we derive reasonable approximations of reachable set in coherently controlled Markovian quantum systems. To this end, we first study the upper bound of system purity function, thus giving an over-approximation that the system can not surpass; and then analyze the small time local controllability, which results in an under-approximation. Moreover, our ideas are implemented experimentally using techniques of nuclear magnetic resonance (NMR).

\emph{Problem setting}---Consider a controlled $n$-qubit open system governed by the Lindblad equation \cite{Lindblad,BP}
\begin{equation}
\label{Lindblad}
\dot \rho   =  - i[H_S + H_C(t),\rho ] +  \mathcal{R}\rho ,
\end{equation}
where $H_S$ is the system Hamiltonian, $H_C(t)$ is the time-dependent external control Hamiltonian, and $\mathcal{R}$ is the relaxation superoperator of Lindblad type.
For simplicity we make two assumptions: (i) relaxation rates are comparatively slow so that arbitrary unitary operation can be implemented before relaxation effects become important; (ii) system's free relaxation process leads to a strictly contractive channel $\mathcal{E}_0$, namely the trace distance of any pair of different states is time decreasing. The latter assumption implies that, there is a unique relaxation-free state $\rho_{eq}$ satisfying: $\mathcal{R} \rho_{eq} = 0$ \cite{NC}. Our assumptions are valid in many practical physical systems that are weakly interacted with a heat bath \cite{BP}, e.g., atoms in a quantized radiation field and spin-lattice systems.

Now introduce the \emph{vector of coherence representation} \cite{SW, K}. Let $\mathcal{B} = \left\{ B_k \right\}_{k=0}^{4^n-1} = \left\{I,X,Y,Z\right\}^{\otimes n}$, where $I$ is the identity, and $X$, $Y$, $Z$ are Pauli operators. It constitutes an orthonormal basis of the state space in that the  orthonormal relation holds: $\operatorname{Tr}\left( B_k B_j \right)/2^n = \delta_{kj}$ for $k, j=0, ..., 4^n-1$.
Consequently $\rho$ can be expressed as: $\rho = {I^{ \otimes n}}/{2^n} + \sum\nolimits_{k = 1}^{{4^n-1}} {{\bm{r}_k}{B_k}} $ (${\bm{r}_k} = \text{Tr}\left( {\rho {B_k}} \right)/2^n$).
The Lindblad Eq. (\ref{Lindblad}) is then turned into a real $4^n -1$ dimensional nonhomogeneous vector differential equation
\begin{equation}
\label{Bloch}
\bm{\dot  r} = \mathbf{H} \bm{r} -  \mathbf{R} ( \bm{r} - {\bm{r}_{eq}}),
\end{equation}
in which $\mathbf{H}$, $\mathbf{R}$  and $\bm{r}_{eq}$ are $4^n-1$ dimensional with their entries determined by ${\mathbf{H}_{kj}}  = \text{Tr} \left( {-i{B_k}\left[H_S+H_C, {{B_j}} \right]} \right)/2^n$, ${\mathbf{R}_{kj}} = \text{Tr}\left(- {{B_k}\mathcal{R}{B_j}} \right)/2^n$ and $\bm{r}_{eq,k}  = \sum\nolimits_j {{\mathbf{R}}_{kj}^{ - 1}{\text{Tr}}\left( {{B_j}{\mathcal R}{I^{ \otimes n}}} \right)/{4^n}} $ respectively. It can be verified that $\mathbf{H}$ is antisymmetric and $\mathbf{R}$ (relaxation matrix) is symmetric positive definite \cite{S}.

To further simplify the problem, we use the diagonalization procedure to project the system dynamics into the diagonal subspace spanned by $\left\{I, Z\right\}^{\otimes n}$ \cite{Y, RBR}. Note that the vector of eigenvalues is essentially $2^n-1$ dimensional. Let $\rho = U \Lambda U^\dag$, where $\Lambda$ is diagonal and $U$ is a unitary operation in $SU(2^n)$.  In the vector of coherence representation the diagonalization procedure can be written as $\bm{r} = \mathbf{U} \bm{x}$, where $\bm{x}$ and $\mathbf{U}$ are the representations of $\Lambda$  and $U$ with respect to basis $\mathcal{B}$ respectively.
Consequently, any evolution of the system can be projected into a continuous trajectory in the diagonal subspace. Substituting the diagonalization procedure into Eq. (\ref{Bloch}), we obtain a $2^n -1$ dimensional dynamical equation \cite{Y}
\begin{equation}
\label{Projection}
\bm{\dot x} = - \left[ {{\mathbf{U}^T}\mathbf{R}\mathbf{U}}\right]_{\mathbf{d}} \bm{x} + \left[{\mathbf{U}^T} \mathbf{R} {\bm{r}_{eq}}\right]_{\mathbf{d}},
\end{equation}
in which the notation $[\cdot]_\mathbf{d}$ denotes the diagonal subspace part of its argument.
Provided that any unitary operation can be performed sufficiently fast compared with the relaxation timescale, we have that: (i) if a diagonal state can be reached, then any state on its unitary orbit can also be  generated; (ii) according to Eq. (\ref{Projection}), the system evolving direction at state $\bm{x}$ can be adjusted to any element of the set $\left\{ \bm{\dot x}_{\mathbf{U}} \left| { \mathbf{U} \in SU(2^n)} \right. \right\}$. Thus it suffices to study the projected dynamics, and we can view $SU(2^n)$ as the admissible control set in place of $H_C(t)$. Let $\operatorname{Reach}_{SU(2^n)}(\rho_{eq}, T)$ $(T \ge 0)$ denote the reachable diagonal states from the equilibrium state $\rho_{eq}$ under the control set $SU(2^n)$ during time $[0,T]$, the global reachable set is defined to be $\bigcup\nolimits_{T \ge 0} {\operatorname{Reach}_{SU(2^n)}(\rho_{eq}, T)}$.

Our goal is thus to construct both over-approximation  and under-approximation of the reachable region of diagonal states. Clearly the problem here extends the concept of \emph{universal bound on spin dynamics} \cite{S90}, i.e., bounds on the regions of operators in Liouville space being interconvertible by unitary transformations, to the open system control regime.
More precisely, let $\rho$ and $\sigma$ be two diagonal states, define the projection with respect to $\sigma$
\begin{equation}
\kappa_U  = \operatorname{Tr}(U \rho U^\dag \cdot \sigma)/\operatorname{Tr}(\sigma^2),
\end{equation}
where $U \in SU(2^n)$, and
\begin{equation}
\kappa_{\mathcal{E}} = \operatorname{Tr}(\mathcal{E}\rho \cdot \sigma)/\operatorname{Tr}(\sigma^2),
\end{equation}
where $\mathcal{E}$ is the non-unitary channel given by Eq. (\ref{Lindblad}) and satisfies $\mathcal{E} \rho = \kappa_{\mathcal{E}} \sigma$. Note that here in the latter case we don't allow the existence of unwanted components.
When $\sigma$ is the target operator, we can interpret $\kappa_U$ and $\kappa_{\mathcal{E}}$ as the polarization transfer efficiency from $\rho$ to $\sigma$. The universal bound gives an analytic expression bounding $\kappa_U$, which bears nice geometric meaning: the unitarily convertible region is bounded by a convex polytope  whose vertexes are composed of all of the diagonal permutations of $\rho$. To move forward a step, we here study the extended problem of determining achievable regions for $\kappa_{\mathcal{E}} $.

\emph{Over-approximation}---
In over-approximating the reachable set, one identifies  regions that the system can never reach. Our approach is to explore the dynamical behaviours of system purity function. Purity, quantifying the incoherent impacts from the environment, is particularly suited for studying how relaxation  imposes restrictions on system evolution. For example, one basic result for \emph{unital} systems, where $\rho_{eq}$ is the maximally mixed state, states that the purity function must be monotonically decreasing with time regardless of the controls \cite{LSA}. For the case of non-unital dynamics the situation is more complicated since purification can occur \cite{A}. However, in practical situations,  purification can not proceed unlimitedly, hence it is natural to seek for an upper bound.

Recall that purity is defined as $p=\text{Tr}\rho^2 = 1/2^n + 2^n {\bm{r}^T} \bm{r}$, thus its first time derivative is given by
\begin{equation}
\dot p =  - 2^{n+1}{\bm{r}^T} \mathbf{R} (\bm{r} - {\bm{r}_{eq}}).
\end{equation}
The set of states satisfying $\dot p = 0$ determines an ellipsoid in $\mathbb{R}^{4^n-1}$, which depends only upon $\mathbf{R}$ and the equilibrium state. From positive definiteness of $\mathbf{R}$ we know that for any state $\bm{r}$ outside of the ellipsoid there must be $\dot p (\bm{r}) < 0$. Let $S$ denote the smallest sphere enclosing the ellipsoid, it is obvious that: (i) the state $\bm{r}_{eq}$ is located inside $S$ and (ii) the evolution direction of any state on $S$ is towards the inner side of $S$. Thus starting at $\bm{r}_{eq}$, the system can not be driven outside $S$ by coherent means. One can then envisage a simple method to get an upper bound of $p$ by solving the following optimization problem
\begin{numcases}{}
\label{Optimization}
\max & $p (\bm{r})= 1/2^n + 2^n{\bm{r}^T} \bm{r}$,  \nonumber \\
\text{s.t.} & $\dot p (\bm{r}) =  - 2^{n+1}{\bm{r}^T} \mathbf{R} (\bm{r} - {\bm{r}_{eq}}) =0.$  \nonumber
\end{numcases}
This problem can be seen as an instance of \emph{quadratic programming over an ellipsoid constraint}, which is easy in the sense of computational complexity and can be solved with well-developed algorithms \cite{FP}.

\emph{Under-approximation}---Under-approximation involves some simplifications of the problem, which we state as such: (i) we restrict our considerations to the discrete set of controls $\mathcal{Q} \subset SU(2^n)$, where $\mathcal{Q}$ is the collection of $2^n!$ permutation operations on diagonal elements of the density matrix; (ii) we will find the small-time local controllable (STLC) set of states rather than analyzing global controllability. The system is said to be \emph{small-time local controllable} at point $\bm{x}$ if $\bm{x}$ belongs to the interior of the reachable set $\operatorname{Reach}_{SU(2^n)}(\bm{x}, T)$ for all $T > 0$.  In other words, for STLC at a point we need to be able to generate small motions in any direction of the full space $\mathbb{R}^{2^n-1}$ at that point. We now use $\Omega_{\mathcal{Q}}$ to denote the STLC set under the discrete control set $\mathcal{Q}$.

The problem of analytically constructing $\Omega_{\mathcal{Q}}$ was studied in full length in Ref. \cite{R}, with the conclusion that $\Omega_{\mathcal{Q}}$ is open, compact and connected, and its boundary is composed of a number of hypersurfaces.
%More precisely, it was proved there that, for every subset $\sigma  \subset \left\{1,2,...,2^n!\right\}$ with $2^n-1$ elements, construct the following hypersurface: $\bm{x}_{\sigma} = {\left( {\sum\nolimits_{k \in \sigma } {{\mu _k}\mathbf{Q}_k^T\mathbf{R_d}{\mathbf{Q}_k}} } \right)^{ - 1}}\left( {\sum\nolimits_{k \in \sigma} {{\mu _k}\mathbf{Q}_k^T\mathbf{R_d}{\bm{x}_{eq}}} } \right)$,
%where ${\mu _k} \ge 0$ and $\sum\nolimits_{k \in \sigma } {{\mu _k}}  = 1$, then  $\Omega_{\mathcal{Q}}$ is an open set whose closure is equal to the closure of $\bigcup\nolimits_\sigma  {{\bm{x}_\sigma }} $.
Knowing about the connectedness of $\Omega_{\mathcal{Q}}$, along with the definition of STLC, we can take $\Omega_{\mathcal{Q}}$ as an under-approximation
\begin{equation}
\Omega_{\mathcal{Q}} \subset \bigcup\nolimits_{T \ge 0} {\operatorname{Reach}_{SU(2^n)}(\rho_{eq}, T)} .
\end{equation}
An algorithmic procedure of calculating $\Omega_{\mathcal{Q}}$ is presented in Supplemental Material \cite{S}.

\emph{Applications on two-qubit system}.---
We use the $^{13}$C-labeled chloroform dissolved in $d_6$-acetone as a two-qubit system to test the applicability of our reachability analysis. Our experiments were carried out on a Bruker Avance \uppercase\expandafter{\romannumeral3} 400 MHz ($B_0$ = 9.4 T) spectrometer at room temperature. The natural Hamiltonian at the basis $\mathcal{B}_2 = \left\{I,X,Y,Z\right\}^{\otimes 2}$ reads: $H_S = \pi ( - \gamma_\text{C} B_0 ZI -  \gamma_\text{H} B_0 IZ+ J/2 ZZ)$,
where $\gamma_\text{C}$ and $\gamma_\text{H}$ are the gyromagnetic ratios of nucleus $^{13}$C and $^{1}$H respectively, and $J = 214.5$Hz is the scalar coupling constant.
The equilibrium state is of the form: $\rho_{eq} \approx II/4 + \epsilon_{\text{C}} ZI + \epsilon_{\text{H}} IZ$ with $\epsilon_{\text{C}} \approx  \epsilon_{\text{H}}/4 \equiv \epsilon  \sim 10 ^{-5}$.
In the double rotating frame, we measured all the effective relaxation rates and thus obtain the system relaxation matrix $\mathbf{R}$ \cite{S}.

\begin{figure*}
\centering
\includegraphics[width=0.95\linewidth]{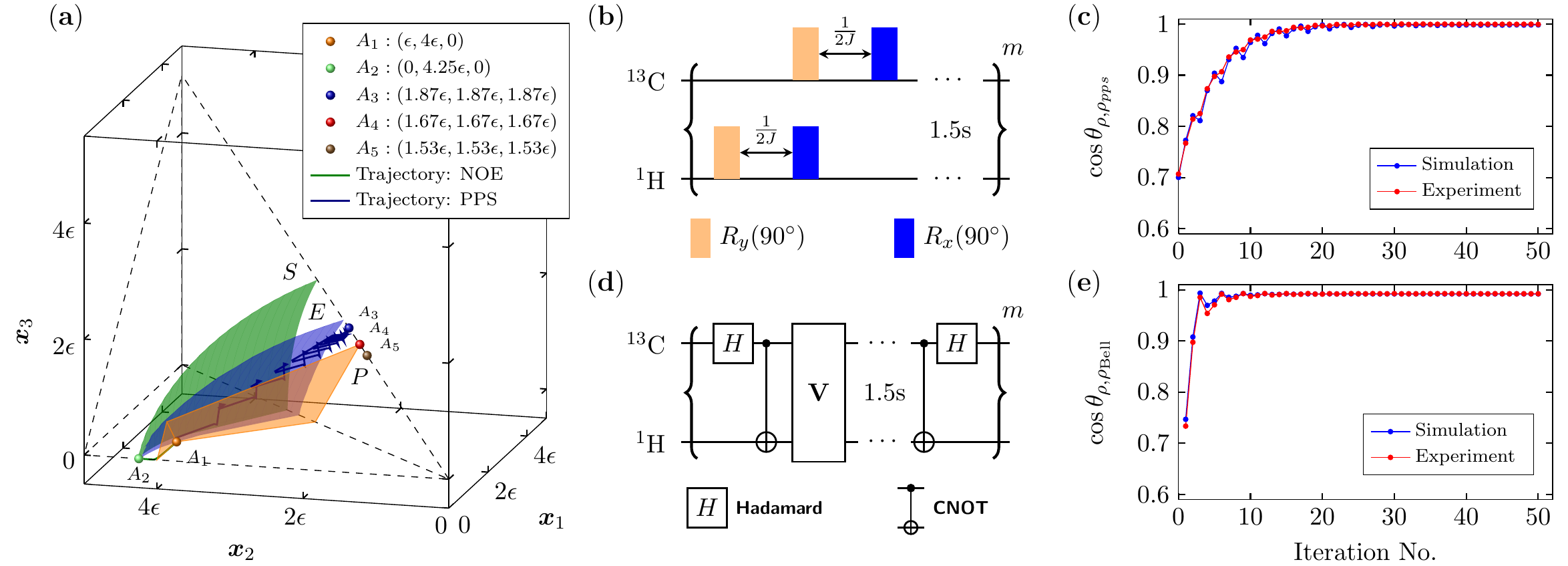}
%\caption{(a) Illustrations of the results (in the region $0 \le \bm{x}_3 \le \bm{x}_1 \le \bm{x}_2$) on our chloroform system, including: (i) sphere (green) $S: \bm{x}^T \bm{x} = 18.06 \epsilon^2$; (ii) surface (blue) $E$: boundary for STLC region under control set $\mathcal{Q}$; (iii) faces (orange) $P$ representing the unitary universal bound; (iv) projected trajectories (simulation) for PPS preparation and NOE experiments. (b)-(c) The resulting PPS $\rho_{pps}$ under periodic control $\left[ 1.5\text{s} - \mathbf{V} \right]_m$, in which $\mathbf{V}$ is implemented through the sequence ${R_y^{\text{H}}(90^\circ) - 1/(2J) - R_x^{\text{H}}(90^\circ) - R_y^{\text{C}}(90^\circ) - 1/(2J) - R_x^{\text{C}}(90^\circ)}$. $\theta_{\rho,\rho_{pps}} $ is the angle between the system state and the PPS direction.  (d)-(e) The resulting Bell state $\rho_{Bell}$ under periodic control $\left[ \mathbf{W} - 1.5\text{s} - \mathbf{V} - \mathbf{W}^T \right]_m$. $\theta_{\rho,\rho_{Bell}} $ is the angle between the system state and the Bell state direction. The data in (c) and (e) clearly display how the system converges to the desired state direction, and numerical simulation with the master equation Eq. (\ref{Bloch}) agrees well with experimental results.
\caption{(a) Illustrations of the results (in the region $0 \le \bm{x}_3 \le \bm{x}_1 \le \bm{x}_2$) on our chloroform system, including: (i) sphere (green) $S: \bm{x}^T \bm{x} = 18.06 \epsilon^2$; (ii) surface (blue) $E$: boundary for STLC region under control set $\mathcal{Q}$; (iii) faces (orange) $P$ representing the unitary universal bound; (iv) projected trajectories (simulation) for PPS preparation and NOE experiments. (b)-(c) The resulting PPS $\rho_{pps}$, and (d)-(e) pseudo-Bell state $\rho_{Bell}$ under periodic controls, along with numerical simulations with the Lindblad equation Eq. (\ref{Bloch}). Here $\theta_{\rho,\rho_{pps}}$ and $\theta_{\rho,\rho_{Bell}}$ denotes the angle between the direction of $\rho$ and that of the desired states. As long as the excitation is on, the system will be preserved periodically at the desired state direction.
}
\label{Figure}
\end{figure*}

In order to visualize the system evolution, we project the 15-dimensional relaxation dynamics into a 3-dimensional differential equation according to Eq. (\ref{Projection}).
For simplicity, we focus on the region $0 \le \bm{x}_3 \le \bm{x}_1 \le \bm{x}_2$ in which our control experiments are performed. In the region, we derived the sphere $S$ representing the upper bound of system purity and the surface $E$ representing the boundary for the system STLC set under the discrete control set $\mathcal{Q}$ based on the measured relaxation matrix (see Fig. \ref{Figure}(a)). The boundary for the exact reachable set should thus be in between $S$ and $E$. We also plotted the faces of the polygon, denoted by $P$, representing the universal bound on spin dynamics under unitary control. It can be seen that, in this system, the exact reachable set exceeds the unitary universal bound in almost every direction, clearly demonstrating the possibility of  larger reachable region of states when relaxation is present. Now we will study how the obtained results help us gain insights into the open system control methods.

Our first concern is the intersection between $S$ and the $\bm{x}_2$ axis: $(0, 4.27\epsilon, 0)$. We can use the nuclear Overhauser effect (NOE) to approach this state. It is well-known that \cite{L}, for a heteronuclear two-spin system,  applying a field at the resonance frequency of one spin for a sufficiently long time, will saturate its polarization and at the meantime affect or even enhance the magnetization of the other spin. In the experiment, an irradiation with $1000$Hz of magnitude and $10$s of duration is applied to the carbon channel, which drives the system into a steady state measured as $\bm{x}_{ss} \approx (0, 4.25\epsilon, 0)$. The $^1$H polarization is enhanced and is fairly close to the upper bound (point $A_2$ in Fig. \ref{Figure}(a)), in contrast to the unitary bound.

Our purity bound analysis thus leads  to a new view of the NOE experiment. The polarization transfer efficiency in NOE experiment is nearly optimal, which shows the advantages of environment-assisted quantum control.   Moreover, it extends the results of algorithmic cooling schemes. According to the purification limits derived in \cite{RMBL}, it is impossible to cool the proton in our system through the ``compression and refresh" iterative procedure. This is due to the different underlying relaxation model assumed. In heat-bath algorithmic cooling scheme, it is considered that each qubit is undergoing their respective $T_1$ and $T_2$ processes. But in NOE, cross-relaxation mechanisms are essential for the $^1$H purification \cite{L}. Thus NOE provides clear evidence of approaching even larger purification efficiency if more general relaxation mechanisms are taken into account.
%Design or evaluation of experimental schemes for trans- fer of coherence or polarization between states in quantized systems requires a clear understanding of which regions of operators in Liouville space are interconvertible by avail- able propagators. Thi
%Overhauser phenomenon is also interesting from the open-system algorithmic cooling aspect of view.

Next we turn to the application of open system coherent control to state engineering  in NMR quantum computation.
We consider creating pseudopure state (PPS) \cite{CPH, GC, KCL} from the equilibrium state, which is  an often used initialization step for subsequent computation. The task can not be done merely with unitary operations. Previous methods of PPS preparation involve different ways of realizing non-unitary operations \cite{PPS} such as exertion of gradient fields. Here, we put forward a new approach: to let the inherent system relaxation effects take the role of non-unitary resources and design a periodic sequence so that PPS is the fixed point of the dynamics. Although the current experiment is performed on two-qubit system as an example, the idea can apply to general cases.

For chloroform, PPS takes the form: $\rho_{pps} = II/4 + \eta/4 (ZI + IZ + ZZ)$, in which $\eta$ is the \emph{effective purity}. The feature that its three coefficients are equal to each other specifies the \emph{PPS direction}, namely $\bm{x}_1 = \bm{x}_2 =\bm{x}_3$. Therefore, it is straightforward to conceive a simple ``coefficient-averaging process". The averaging process is governed by $\left[ \tau - \mathbf{V} \right]_m$, where $m$ is the iteration number, $\tau$ represents a period of free relaxation, and $\mathbf{V}$ is a cyclic permutation of the coordinates of $\bm{x}$. Let $\mathcal{E}_\tau$ and $\mathcal{E}_\mathbf{V}$ denote the dynamic map associated with the $\tau$ relaxation evolution and $\mathbf{V}$ operation respectively. Provided that $\tau$ is small enough, namely $\mathcal{E}_\tau$ is close to the identity, then the fixed point of $\mathcal{E}_\mathbf{V} \circ \mathcal{E}_{\tau}$ is close to that of $\mathcal{E}_\mathbf{V}$.
In the experiments, we chose $\mathbf{V}:  (\bm{x}_1,\bm{x}_2,\bm{x}_3)^T \to  (\bm{x}_2,\bm{x}_3,\bm{x}_1)^T$, which is implemented through a simple sequence
%: $R_y^{\text{H}}(90^\circ) - 1/(2J) - R_x^{\text{H}}(90^\circ) - R_y^{\text{C}}(90^\circ) - 1/(2J) - R_x^{\text{C}}(90^\circ)$
shown in Fig. \ref{Figure}(b). It was found that for a wide range of $\tau$ (less than $\sim 2$s) the system was able to be driven to some states close to the PPS direction. Within tolerable range of error, we set $\tau = 1.5$s, giving the  maximal effective purity ($\eta \approx 7.48 \epsilon$, point $A_3$ in Fig. \ref{Figure}(a)) of PPS obtained on trials. This can be compared to conventional spatial averaging method where $\eta \approx 6.12 \epsilon $ (point $A_5$ in Fig. \ref{Figure}(a)) \cite{P}  and line-selective pulse approach where unitary bound can be achieved $\eta \approx 6.67 \epsilon $ (point $A_4$ in Fig. \ref{Figure}(a)) \cite{P01}. Fig. \ref{Figure}(a) also reveals a gap between the point of maximally reachable $\eta$ (in between over-approximation and under-approximation) and the obtained PPS.
The  gap can be attributed to several reasons: (i) the imprecision in the experimental estimation of the relaxation matrix $\mathbf{R}$; (ii) the assumption of ignoring relaxation effects during the operation $\mathbf{V}$ is not perfectly satisfied in practice; (iii) the over-approximation may be not sufficiently tight and there is also the possibility that better preparation method exists.

Our periodic control method applies to any state that is unitarily equivalent to a PPS, e.g., a pseudo-Bell state $\rho_{Bell} = (1-\eta)/4 II + \eta/2 \left( {\left| {{\rm{00}}} \right\rangle {\rm{ + }}\left| {{\rm{11}}} \right\rangle } \right) \otimes \left( {\left\langle {{\rm{00}}} \right|{\rm{ + }}\left\langle {{\rm{11}}} \right|} \right)$. The trick is that, we just modify the PPS preparation periodic sequence to be $\left[ \mathbf{W} - \tau - \mathbf{V} - \mathbf{W}^T \right]_m$ (Fig. \ref{Figure}(d)), so that $\rho_{Bell}$ now becomes the fixed point of the sequence. Here $\mathbf{W}$ transforms $\rho_{pps}$ to $\rho_{Bell}$ and can be implemented through a Hadamard gate and a CNOT gate.
The experimental results shown in Fig. \ref{Figure}(e) demonstrates that the proposed control method can be implemented and gives results in excellent agreement with predictions.

The theory presented here helps assessing open system control schemes where purity is an important metric.
In addition, our  approximation can guide the development of numerical pulse searching algorithms. Our work can be improved by increasing the efficiency of computing the approximation, with the aid of advanced algorithmic techniques from computational geometry. We further studied in detail the NOE effect and state engineering experiments in the open system framework, and showed that relaxation effects are useful for implementing some nontrivial non-unitary control tasks. The lack of full controllability in certain important control regimes \cite{XYS, L} usually calls for a bound analysis for system reachable states. Our present study can thus be regarded as a part of explorations in this direction.  Future work will concentrate on incorporating our work here with other open system control models \cite{DP}, such as reservoir engineering in which incoherent resources are introduced to enhance the capability of controlling quantum systems.

\section{Acknowledgments}
This work is supported by the National Key Basic Research Program of China (Grant No. 2013CB921800 and No. 2014CB848700), the National
Science Fund for Distinguished Young Scholars Grant No. 11425523, National Natural Science Foundation of China under Grant Nos. 11375167, 11227901, 91021005,
the Chinese Academy of Sciences, the Strategic Priority Research Program (B) of the CAS (Grant No. XDB01030400), and Research Fund for the Doctoral Program of Higher Education of China under Grant No. 20113402110044.

\newpage
\onecolumngrid
\appendix

\section{Lindblad Equation in Vector of Coherence Representation}
Here we rewrite Lindblad equation in the vector of coherence representation, basically following the discussions in \cite{A,K,SW}. Introducing the orthonormal basis $\mathcal{B} = \left\{ B_k \right\}_{k=0}^{4^n-1} = \left\{I,X,Y,Z\right\}^{\otimes n}$, then there is
\begin{equation}
\rho = {I^{ \otimes n}}/{2^n} + \sum\limits_{k = 1}^{{4^n-1}} {{\bm{r}_k}{B_k}},  \quad {\bm{r}_k} = \text{Tr}\left( {\rho {B_k}} \right)/2^n.
\end{equation}
Substituting the above expression into the Lindblad Eq. (\textcolor[rgb]{0.00,0.00,1.00}{1}) yields
\begin{align}
{\bm{\dot r}_k}
& = \text{Tr}\left( {{B_k}\left( { - i\left[ {H,\rho } \right] + \mathcal{R}\rho } \right)} \right)/2^n \nonumber \\
& = \text{Tr}\left( { - i{B_k}\left[ {H,\sum\limits_{j = 1}^{{4^n} - 1} {{\bm{r}_j}{B_j}} } \right]} \right)/2^n + \text{Tr}\left( {{B_k}\mathcal{R}\left( {\frac{{{I^{ \otimes n}}}}{{{2^n}}} + \sum\limits_{j = 1}^{{4^n} - 1} {{\bm{r}_j}{B_j}} } \right)} \right)/2^n \nonumber \\
& = \sum\limits_{j = 1}^{{4^n} - 1} {\text{Tr}\left( { - i{B_k}\left[ {H,{B_j}} \right]} \right)/2^n{\bm{r}_j}}  + \sum\limits_{j = 1}^{{4^n} - 1} {\text{Tr}\left( {{B_k}\mathcal{R}{B_j}} \right)/2^n{\bm{r}_j}}  + \text{Tr}\left( {{B_k}\mathcal{R}\left( {\frac{{{I^{ \otimes n}}}}{{{2^n}}}} \right)} \right)/2^n.
\label{1}
\end{align}
Define
\begin{align}
\mathbf{H}: \quad & {\mathbf{H}_{kj}} = \text{Tr}\left( { - i{B_k}\left[ {H,{B_j}} \right]} \right)/2^n, \\
\mathbf{R}: \quad & {\mathbf{R}_{kj}} = \text{Tr}\left( - {{B_k}\mathcal{R}{B_j}} \right)/2^n,  \\
\bm{v}: \quad & \bm{v}_k = \text{Tr}\left( {{B_k}\mathcal{R}{I^{ \otimes n}}} \right)/{4^n}.
\end{align}
Then we get
\begin{equation}
\bm{\dot  r} = \mathbf{H} \bm{r} -  \mathbf{R} \bm{r} + \bm{v}.
\end{equation}

\begin{proposition}
$\mathbf{R}$ is real, symmetric. If the pure relaxation process of the system is strictly contractive, then $\mathbf{R}$ is positive definite.
\end{proposition}

\begin{proof}
In general $\mathbf{R}$ is self-adjoint, the proof of which can be found in \cite{K,A03}. Provided that $\rho$ is decomposed with respect to the generalized Pauli basis $\left\{ B_k \right\}_{k=0}^{4^n-1}$, then as $\mathcal{R} B_j$ is self-adjoint, $\text{Tr}\left(- {{B_k}\mathcal{R}{B_j}} \right)$ should be real for all $k, j =1, ..., 4^n-1$. So $\mathbf{R}$ is real and hence symmetric.

If the pure relaxation process leads to a strictly contractive channel, then there exists a unique fixed point, which we denote by  $\bm{r}_0$ (\cite{NC}, page 408). It is the unique solution to the equation $(\mathbf{H}_S - \mathbf{R}) \bm{r}_0 + \bm{v} = 0$. Thus $\mathbf{H}_S - \mathbf{R}$ should be of full rank such that $\bm{r}_0 = (-\mathbf{H}_S + \mathbf{R})^{-1}\bm{v}$. The pure relaxation dynamics is then described by
\begin{equation}
\bm{\dot  r} = (\mathbf{H}_S -  \mathbf{R}) (\bm{r} -\bm{r}_0). \nonumber
\end{equation}
Furthermore, the property of strictly contractiveness implies that, for any state $\bm{r}(t)$ other than $\bm{r}_0$, the time derivative of the trace distance of which to $\bm{r}_0$ should satisfy
\begin{equation}
\dot D(\bm{r}(t), \bm{r}_0) =  \frac{1}{2} \frac{d \left| \bm{r} - \bm{r}_0 \right|}{dt}
= -   \frac{(\bm{r} - \bm{r}_0)^T \mathbf{R} (\bm{r} - \bm{r}_0)}{4 \left| \bm{r} - \bm{r}_0 \right|}
<  0. \nonumber
\end{equation}
This condition can be fulfilled only if $\mathbf{R}$ is symmetric positive definite.
\end{proof}

Since we have assumed at the beginning that our considered system under pure relaxation process is strictly contractive, $\mathbf{R}$ should be of full rank by the above proposition, thus if we define
\begin{equation}
\bm{r}_{eq} = \mathbf{R}^{-1} \bm{v},
\end{equation}
then accordingly we will have
\begin{equation}
\bm{\dot  r} = \mathbf{H} \bm{r} -  \mathbf{R}(\bm{r} - {\bm{r}_{eq}}).
\label{bloch}
\end{equation}

\section{Derivation of Eq. (4)}
The derivation of the projected dynamic equation can actually be found in \cite{Y}, which we copy as follows.

At each instant of time, we can diagonalize the density matrix $\rho(t) = U(t) \Lambda(t) U^\dag (t)$ by a unitary matrix $U(t)$. Substitute it into the Lindblad Eq. (1), we get
\begin{align}
\dot \Lambda (t) & = {{\dot U}^\dag }(t)\rho (t)U(t) + {U^\dag }(t)\dot \rho (t)U(t) + {U^\dag }(t)\rho (t)\dot U(t)  \nonumber \\
& =  i{U^\dag }(t){H'}(t)\rho (t)U(t) + {U^\dag }(t)\left\{ { - i\left[ {H(t),\rho(t) } \right] + \mathcal{R}\rho(t) } \right\}U(t) - {U^\dag }(t)\rho (t)i{H'}(t)U(t)  \nonumber \\
& =  - i{U^\dag }(t)\left[ {H(t) - H'(t),\rho(t)} \right]U(t) + {U^\dag }(t)\mathcal{R}\left( {\rho(t)} \right)U(t),  \nonumber \\
& =  - i{U^\dag }(t)\left[ {H(t) - H'(t),U(t)\Lambda (t){U^\dag }(t)} \right]U(t) + {U^\dag }(t)\mathcal{R}\left( {U(t)\Lambda (t){U^\dag }(t)} \right)U(t),  \nonumber
\end{align}
where we have defined $\dot U(t) =  - i{H'}(t)U(t)$, and ${H'}(t) $ must be Hermitian by the fact that  $d(U(t) U^\dag(t))/dt = 0$.

Note that the left side of the above equation is a diagonal matrix, so for the right side we only need to keep the diagonal part. Moreover, the first term on the right side is a commutation of two Hermitian matrices, and since $\Lambda (t)$ is diagonal, so the diagonal part of this commutation must be zero. Therefore, the above equation reduces to
\begin{equation}
\dot \Lambda (t) = \operatorname{diag} \left( {U^\dag }(t)\mathcal{R}\left( {U(t)\Lambda (t){U^\dag }(t)} \right)U(t)
 \right).
\end{equation}

Now we go to the vector representation.

\begin{center}
\begin{tikzpicture}[every node/.style={midway},scale=2]
  \matrix[column sep={4em,between origins}, row sep={2em}] at (0,0) {
    \node(rho) {$\rho$}  ; & \node(Lambda) {$\Lambda$}; \\
    \node(r) {$\bm{r}$}; & \node (x) {$\bm{x}$};\\
  };
  \draw[->] (rho) -- (r) ;
  \draw[->] (Lambda) -- (rho) node[anchor=south]  {\footnotesize $U(t)$};
  \draw[->] (Lambda) -- (x);
  \draw[->] (x) -- (r) node[anchor=north] {\footnotesize $\mathbf{U}(t)$};
\end{tikzpicture}
\end{center}
Substitute $\bm{r} = \mathbf{U}(t) \bm{x}$ into Eq. (\ref{bloch}), there is
\begin{align}
\bm{\dot  x} & = {\mathbf{U}^T}\mathbf{H}\mathbf{U}\bm{x} - {\mathbf{U}^T}\mathbf{R}\mathbf{U}\bm{x} + {\mathbf{U}^T}\mathbf{R}{\bm{r}_{eq}} \nonumber \\
& = \left[ {\mathbf{U}^T}\mathbf{H}\mathbf{U}\bm{x} - {\mathbf{U}^T}\mathbf{R}\mathbf{U}\bm{x} + {\mathbf{U}^T}\mathbf{R}{\bm{r}_{eq}} \right]_{\mathbf{d}}. \nonumber
\end{align}
As has been just demonstrated, the first term of the above equation should vanish, thus
\begin{equation}
\bm{\dot x} = - \left[ {{\mathbf{U}^T}\mathbf{R}\mathbf{U}}\right]_{\mathbf{d}} \bm{x} + \left[{\mathbf{U}^T} \mathbf{R} {\bm{r}_{eq}}\right]_{\mathbf{d}},
\end{equation}

\section{Under-approximation of the Reachable Set}
The projected dynamics in the diagonal subspace goes
\begin{equation}
\label{Projection}
\bm{\dot x} = - \left[ {{\mathbf{U}^T}\mathbf{R}\mathbf{U}}\right]_{\mathbf{d}} \bm{x} + \left[{\mathbf{U}^T} \mathbf{R} {\bm{r}_{eq}}\right]_{\mathbf{d}},
\end{equation}
in which $\mathbf{U}$ runs over all elements of the group $SU(2^n)$. For each $\mathbf{U} \in SU(2^n)$, there corresponds to an evolving direction. Since we assume that any unitary operation can be performed very fast compared with the relaxation timescale, the system evolving direction at state $\bm{x}$ can be adjusted to any element of the set
\begin{equation}
\left\{ \bm{\dot x}_{\mathbf{U}} = - \left[ {{\mathbf{U}^T}\mathbf{R}\mathbf{U}}\right]_{\mathbf{d}} \bm{x} + \left[{\mathbf{U}^T} \mathbf{R} {\bm{r}_{eq}}\right]_{\mathbf{d}} \left| { \mathbf{U} \in SU(2^n)} \right. \right\}.
\end{equation}

Now we want to under-approximate the system reachable set. To this end, we study the simplified reachability problem: (i) instead of considering the whole control set $SU(2^n)$, we restrict our attention to the discrete set of controls $\mathcal{Q}$; (ii) we will find the small-time local controllable set of states rather than analyzing global controllability.

For system (\ref{Projection}) under the discrete control set $\mathcal{Q}$,  we denote the set of admissible evolving directions at an arbitrary state $\bm{x}$  as
\begin{equation}
\left\{ \bm{\dot x}_{\mathbf{Q}_k} = - \left[ {{\mathbf{Q}^T}\mathbf{R}\mathbf{Q}}\right]_{\mathbf{d}} \bm{x} + \left[{\mathbf{Q}^T} \mathbf{R} {\bm{r}_{eq}}\right]_{\mathbf{d}}   \left| { \mathbf{Q}_k \in \mathcal{Q}} \right. \right\}.
\end{equation}
Let $\operatorname{cone}(\left\{\bm{\dot x}_{\mathbf{Q}_k}\right\})$ be the convex cone generated by the vector fields $\left\{\bm{\dot x}_{\mathbf{Q}_k}\right\}$ through conical combination:
\begin{equation}
\operatorname{cone}(\left\{\bm{\dot x}_{\mathbf{Q}_k}\right\}) = \left\{  \sum\limits_{k = 1}^{{2^n}!} {{c_k}{{\bm{\dot x}}_{{\mathbf{Q}_k}}}} | {{c_k} \ge 0,{\mathbf{Q}_k} \in \mathcal{Q}}  \right\}.
\end{equation}
Then (\cite{Rooney12}, page 56),
\begin{proposition}
Given the discrete set of admissible vector fields $\left\{\bm{\dot x}_{\mathbf{Q}_k}\right\}$, one can and only can generate motions in the convex cone $\operatorname{cone}(\left\{\bm{\dot x}_{\mathbf{Q}_k}\right\})$.
\end{proposition}

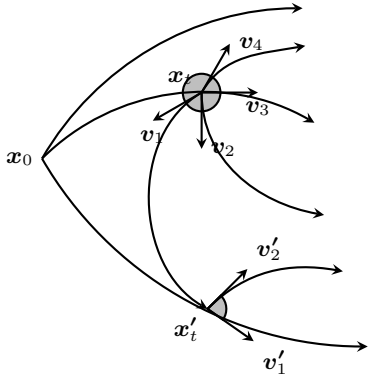
\begin{figure}[h]
\centering
 \begin{tikzpicture}
 %\node at (-0.5,2.5) {(\textbf{a})};

 \draw[thick,fill=gray!50] (3/1.414,3-3/1.414) ellipse (0.25cm and 0.25cm);
 \draw[thick,fill=gray!50] (2.2,-2) -- (2.2+0.25*0.8192,-2-0.25*0.5736) arc (-35:45:0.25cm)-- cycle;

 \draw [>=stealth,->,thick] (0,0) arc [radius=4, start angle=150, end angle= 90];
 \draw [>=stealth,->,thick] (0,0) arc [radius=5, start angle=210, end angle= 270];

 \draw [>=stealth,->,thick] (0,0) arc [radius=3, start angle=135, end angle= 60];
 \draw [>=stealth,->,thick] (3/1.414,3-3/1.414) to [out=60,in=190] (3.5,1.5);
 \draw [>=stealth,->,thick] (3/1.414,3-3/1.414) to [out=-90,in=170] (3.75,-0.75);
 \draw [>=stealth,->,thick] (3/1.414,3-3/1.414) to [out=-150,in=150] (2.2,-2);

 \draw [>=stealth,->,thick] (3/1.414,3-3/1.414) -- (3/1.414+0.75*0.5,3-3/1.414+0.75*0.866) node [anchor=west] {\small $\bm{v}_4$};
 \draw [>=stealth,->,thick] (3/1.414,3-3/1.414) -- (3/1.414+0.75,3-3/1.414) node [anchor=north] {\small $\bm{v}_3$};
 \draw [>=stealth,->,thick] (3/1.414,3-3/1.414) -- (3/1.414,3-3/1.414-0.75) node [anchor=  west] {\small $\bm{v}_2$};
 \draw [>=stealth,->,thick] (3/1.414,3-3/1.414) -- (3/1.414-0.75*0.866,3-3/1.414-0.75*0.5) node [anchor=north] {\small $\bm{v}_1$};

 \draw [>=stealth,->,thick] (2.2,-2) to [out=45,in=170] (4,-1.5);
 \draw [>=stealth,->,thick] (2.2,-2) -- (2.2+0.75/1.414,-2+0.75/1.414) node [anchor=south west] {\small $\bm{v'}_2$};
 \draw [>=stealth,->,thick] (2.2,-2) -- (2.2+0.75*0.8192,-2-0.75*0.5736) node [anchor=north west] {\small $\bm{v'}_1$};

 \node [anchor=east] at (0,0) {\small $\bm{x}_0$};
 \node [anchor=south east] at (3/1.414,3-3/1.41) {\small $\bm{x}_t$};
 \node [anchor=north east] at (2.2,-2) {\small $\bm{x'}_t$};

 %\node at (6,2.5) {(\textbf{b})};
\end{tikzpicture}
\caption{Illustration of STLC property. It can be seen that the cone generated by the velocity vectors  at $\bm{x}_t$ is the full space, while this is not true for $\bm{x'}_t$.}
\end{figure}

Let $\operatorname{Reach}_{\mathcal{Q}}(\bm{x}, T)$ $(T>0)$ denote the reachable set from state $\bm{x}$ under control $\mathcal{Q}$ during time $[0,T]$. We say system (\ref{Projection}) is \emph{small-time local controllable (STLC) at point $\bm{x}$} if $\bm{x}$ belongs to the interior of the reachable set $\operatorname{Reach}_{\mathcal{Q}}(\bm{x}, T)$ for all $T > 0$. In other words, for STLC at a point we need to be able to generate small motions in any direction of the full space $\mathbb{R}^{2^n-1}$ at that point. So one has that,
system (\ref{Projection}) is STLC at point $\bm{x}$ iff $\operatorname{cone}(\left\{\bm{\dot x}_{\mathbf{Q}_k}\right\}) = \mathbb{R}^{2^n-1}$. We denote the system STLC set under the discrete control set $\mathcal{Q}$ by $\Omega_{\mathcal{Q}}$.

\subsection{Constructing the STLC Set}

P. Rooney analytically constructed the STLC set under the control set $\mathcal{Q}$ in Ref. \cite{Rooney12}. It turns out that $\Omega_{\mathbf{Q}}$ is open, compact and connected, and its boundary is composed of a number of surfaces. We here just copy the core result obtained by P. Rooney:

\begin{theorem}[\cite{Rooney12}, page 94]
For every subset $\sigma  \subset \left\{1,2,...,2^n!\right\}$ with $2^n-1$ elements, construct the hypersurface $\bm{x}_{\sigma} = {\left( {\sum\nolimits_{k \in \sigma } {{\mu _k}\left[\mathbf{Q}_k^T\mathbf{R}{\mathbf{Q}_k} \right]_{\mathbf{d}} }} \right)^{ - 1}}\left( {\sum\nolimits_{k \in \sigma} {{\mu _k}\left[\mathbf{Q}_k^T\mathbf{R}{\bm{r}_{eq}}\right]_{\mathbf{d}} } } \right)$, where ${\mu _k} \ge 0$ and $\sum\nolimits_{k \in \sigma } {{\mu _k}}  = 1$. Denote $\bigcup\nolimits_\sigma  {{\bm{x}_\sigma }} $ as the union of all such hypersurfaces, which is a closed hypersurface.  Then $\Omega_{\mathbf{Q}}$ is an open set whose closure is equal to the closure of $\bigcup\nolimits_\sigma  {{\bm{x}_\sigma }} $.
\end{theorem}

\subsection{Algorithm for Testing STLC}
In practice, it  is extremely difficult to compute   $\Omega_{\mathbf{Q}}$. Just take two-qubit system as an example, the number of surfaces to be computed is $C_{24}^4 = 10626$, and an extra effort of making the union of these surfaces has to be made. Here
in order to compute the boundary of $\Omega_{\mathbf{Q}}$ in the case of a two-qubit system, we choose an alternative approach that can  reduce practical computing efforts to a large extent. Basically, we sample the state space with a discrete set of points, and choose those STLC points among them through a STLC testing algorithm.

%Although Lemma \ref{MainLemma} gives a characterization of  STLC property, it does not provide an algorithmic procedure to test STLC.
To determine whether a vector is in a convex cone, belongs to the class of point-in-polygon problems, and  there exist a number of algorithms to solve them in the realm of computational geometry.
Here we  present a testing algorithm based on  the so called \emph{fundamental theorem of linear inequalities} (this classic theorem is due to Farkas, Minkowski, Carath\'{e}odory, Weyl, etc.) \cite{S86}.

%\begin{proposition}
%\label{Observation}
%Let $\operatorname{Ray}( \bm{\hat x}) = \left\{s \bm{ \hat x}: \bm{\hat x} \in \Sigma, \left\|\bm{\hat x}\right\| = 1,  s \ge 0  \right\}$ be an arbitrary ray in $\Sigma$ starting from the origin. Then there exists just one intersection point between $\operatorname{Ray}( \bm{\hat x})$ and the boundary of $\mathcal{S}_{\mathbf{Q}}$. In other words, for each direction $\bm{\hat x}$ in $\Sigma$, the STLC set on $\operatorname{Ray}( \bm{\hat x})$ is a non-empty continuous interval starting from the origin.
%\end{proposition}
%
%The above proposition suggests that for computing the STLC set $\mathcal{S}_{\mathbf{Q}}$, it suffices to find the vector with minimum norm that is not STLC at each direction. To prove the proposition, we need the Farkas' lemma, which is a classic result in convex analysis:
%\begin{lemma}[Farkas, Geometric Version \cite{G10}]
%Let $\left\{ \bm{v}_k \right\}$ and $\bm{u}$ be vectors in $\mathbb{R}^m$. Then exactly one of the two statements is true: (i) $\bm{u} \in \operatorname{cone}(\left\{ \bm{v}_k \right\})$; (ii) there exists a vector $\bm{n} \in \mathbb{R}^m$ such that $\bm{n} \cdot \bm{v}_k \le 0$ for all $k$ and $\bm{n} \cdot \bm{u} > 0$, which means that there is a hyperplane (with $\bm{n}$ being its normal vector) separating $\operatorname{cone}(\left\{ \bm{v}_k \right\})$ and $\bm{u}$.
%\end{lemma}

\begin{theorem}[Fundamental Theorem of Linear Inequalities]
Let $\left\{ \bm{v}_k \right\}$ and $\bm{u}$ be vectors in $\mathbb{R}^m$, and suppose $\operatorname{span}(\left\{ \bm{v}_k \right\}) =\mathbb{R}^m$. Then exactly one of the two statements is true: (i) $\bm{u} \in \operatorname{cone}(\left\{ \bm{v}_k \right\})$; (ii) there exists a hyeperplane $\left\{ \bm{x} \vert \bm{n} \cdot \bm{x} =0 \right\}$, containing $m-1$ linearly independent vectors from $\left\{ \bm{v}_k \right\}$, such that $\bm{n} \cdot \bm{u} >0$ and  $\bm{n} \cdot \bm{v}_k \le 0$ for all $k$, which means that there is a hyperplane spanned by $m-1$ vectors from $\left\{ \bm{v}_k \right\}$ (with $\bm{n}$ being its normal vector) separating $\operatorname{cone}(\left\{ \bm{v}_k \right\})$ and $\bm{u}$.
\end{theorem}

The theorem provides a criterion to test whether a given set of vector fields is STLC.

\begin{corollary}
Let $\left\{ \bm{v}_k \right\}$ be vectors in $\mathbb{R}^m$, and $\operatorname{span}(\left\{ \bm{v}_k \right\}) =\mathbb{R}^m$. Then  $\operatorname{cone}(\left\{ \bm{v}_k \right\}) = \mathbb{R}^m $, iff for any hyeperplane spanned by $m-1$ linearly independent vectors from $\left\{ \bm{v}_k \right\}$, its normal vector $\bm{n}$ satisfies that $\left\{ \bm{n} \cdot \bm{v}_k\right\}$ is not all nonpositive and not all nonnegative.
\end{corollary}

\begin{proof}
Positive direction. It is evident from the ``fundamental theorem of linear inequalities".

Inverse direction. If there exists a nonzero vector $\bm{n}$ such that  $\left\{ \bm{n} \cdot \bm{v}_k\right\}$ is  all nonpositive or  all nonnegative,  then surely either $\bm{n}$ or $- \bm{n}$ can not be written as a conical combination of $\left\{ \bm{v}_k \right\}$, which means  $\operatorname{span}(\left\{ \bm{v}_k \right\})$ can not be the full space.
\end{proof}

One has consequently Algorithm \ref{STLC} to test the STLC property of the vector fields $\left\{\bm{\dot x}_{\mathbf{Q}_k}\right\}$ at a given state $\bm{x}$.

\begin{algorithm}
\caption{Algorithm for STLC Testing}
\label{STLC}
\begin{algorithmic}[1]
\Require State $\bm{x}$.
\Ensure True if $\bm{x}$ is STLC; False if $\bm{x}$ is not STLC.
\State  Calculate all the vectors $\bm{\dot x}_{\mathbf{Q}_k}$, $k=1, ..., 24$;
\For{$i =1, ..., 23$}
  \For{$j = i+1, ..., 24$}
    \For{$k = 1, ..., 24 \wedge k \neq i,j$}
      \State $c_k = \bm{\dot x}_{\mathbf{Q}_i} \times \bm{\dot x}_{\mathbf{Q}_j} \cdot \bm{\dot x}_{\mathbf{Q}_k}$;
    \EndFor
    \If{$(c_1, ..., c_{24} \ge 0) \vee (c_1, ..., c_{24} \le 0)$}
      \State $d_{ij} =$ False;
    \Else
      \State $d_{ij} =$ True;
    \EndIf
  \EndFor
\EndFor
\State \textbf{Return} $\mathop  \bigwedge \nolimits_{i,j = 1}^{24} {d_{ij}}$.
\end{algorithmic}
\end{algorithm}

\section{Relaxation Matrix Tomography on Chloroform}
The basic liquid NMR relaxation theory can be found in \cite{KM}. In the rotating frame, it is routine to make a secular approximation by which the relaxation matrix would take a kite-like appearance.
The underlying principle is that, the system energy level differences are much larger than the relaxation rates, so in the interaction picture, the cross relaxation parameters between the population subspace and the coherence subspace are added with fast oscillating phases. This effectively decoupled the population subspace relaxation from the coherence subspace relaxation.
Secular approximation dramatically simplified the task of experimentally estimating the relaxation rates.

To be concrete, the relaxation dynamics can be decomposed as a direct sum of subspace dynamics (featured by the order of coherences)
\begin{itemize}
\item{population subspace:
\begin{equation}
\frac{d}{{dt}}\left( {\begin{array}{*{20}{c}}
   {1/4}  \\
   {\left\langle {ZI} \right\rangle }  \\
   {\left\langle {IZ} \right\rangle }  \\
   {\left\langle {ZZ} \right\rangle }  \\
\end{array}} \right) = \left[ {\begin{array}{*{20}{c}}
   0 & 0 & 0 & 0  \\
   { - (4{r_1} + 16{r_4})\varepsilon } & {{r_1}} & {{r_4}} & {{r_5}}  \\
   { - (4{r_4} + 16{r_2})\varepsilon } & {{r_4}} & {{r_2}} & {{r_6}}  \\
   { - (4{r_5} + 16{r_6})\varepsilon } & {{r_5}} & {{r_6}} & {{r_3}}  \\
\end{array}} \right]\left( {\begin{array}{*{20}{c}}
   {1/4}  \\
   {\left\langle {ZI} \right\rangle }  \\
   {\left\langle {IZ} \right\rangle }  \\
   {\left\langle {ZZ} \right\rangle }  \\
\end{array}} \right),
\end{equation}
}
\item{$^{13}$C one-quantum coherence subspace:
\begin{equation}
\frac{d}{{dt}}\left( {\begin{array}{*{20}{c}}
   {\left\langle {XI} \right\rangle }  \\
   {\left\langle {YI} \right\rangle }  \\
   {\left\langle {XZ} \right\rangle }  \\
   {\left\langle {YZ} \right\rangle }  \\
\end{array}} \right) = \left[ {\begin{array}{*{20}{c}}
   r_7 & 0 & r_9 & -\pi J  \\
   0 & r_7 & \pi J & r_9  \\
   r_9 & -\pi J & r_8 & 0  \\
   \pi J & {{r_9}} & 0 & r_8  \\
\end{array}} \right]\left( {\begin{array}{*{20}{c}}
  {\left\langle {XI} \right\rangle }  \\
   {\left\langle {YI} \right\rangle }  \\
   {\left\langle {XZ} \right\rangle }  \\
   {\left\langle {YZ} \right\rangle }  \\
\end{array}} \right),
\end{equation}
}
\item{H one-quantum coherence subspace:
\begin{equation}
\frac{d}{{dt}}\left( {\begin{array}{*{20}{c}}
  {\left\langle {IX} \right\rangle }  \\
   {\left\langle {IY} \right\rangle }  \\
   {\left\langle {ZX} \right\rangle }  \\
   {\left\langle {ZY} \right\rangle }  \\
\end{array}} \right) = \left[ {\begin{array}{*{20}{c}}
   r_{10} & 0 & r_{12} & -\pi J  \\
   0 & r_{10} & \pi J & r_{12}  \\
   r_{12} & -\pi J & r_{11} & 0  \\
   \pi J & {{r_{12}}} & 0 & r_{11}  \\
\end{array}} \right]\left( {\begin{array}{*{20}{c}}
  {\left\langle {IX} \right\rangle }  \\
   {\left\langle {IY} \right\rangle }  \\
   {\left\langle {ZX} \right\rangle }  \\
   {\left\langle {ZY} \right\rangle }  \\
\end{array}} \right),
\end{equation}
}
\item{Subspace of zero- and double- quantum coherences:
\begin{equation}
\frac{d}{{dt}}\left( {\begin{array}{*{20}{c}}
   {\left\langle {XY} \right\rangle }  \\
   {\left\langle {YX} \right\rangle }  \\
   {\left\langle {XX} \right\rangle }  \\
   {\left\langle {YY} \right\rangle }  \\
\end{array}} \right) = \left[ {\begin{array}{*{20}{c}}
  r_{13} & -r_{14} & 0 & 0  \\
   -r_{14} & r_{13} & 0 & 0  \\
   0 & 0 & r_{13} & r_{14}  \\
   0 & 0 & r_{14} & r_{13}  \\
\end{array}} \right]\left( {\begin{array}{*{20}{c}}
   {\left\langle {XY} \right\rangle }  \\
   {\left\langle {YX} \right\rangle }  \\
   {\left\langle {XX} \right\rangle }  \\
   {\left\langle {YY} \right\rangle }  \\
\end{array}} \right),
\end{equation}
}
\end{itemize}
where $\left\{r_k \right\}_{k=1,...,14}$ are relaxation rates including auto-relaxation rates and cross-relaxation rates. To estimate the relaxation rates, we first sample the system evolution trajectory (starting from a known initial state $\rho (0)$), then find values of the relaxation rates so that the simulated dynamics can match the observed trajectory. The fitting results are listed below (we have set $\epsilon = 1$)
\begin{itemize}
\item{$\left\{ r_1, r_2, r_3, r_4, r_5, r_6 \right \} \approx \left\{ 0.0532, 0.0918, 0.0798, 0.0212, 0.0000, 0.0022 \right \}$
\begin{figure}[h]
\centering
\includegraphics[width=0.9\linewidth]{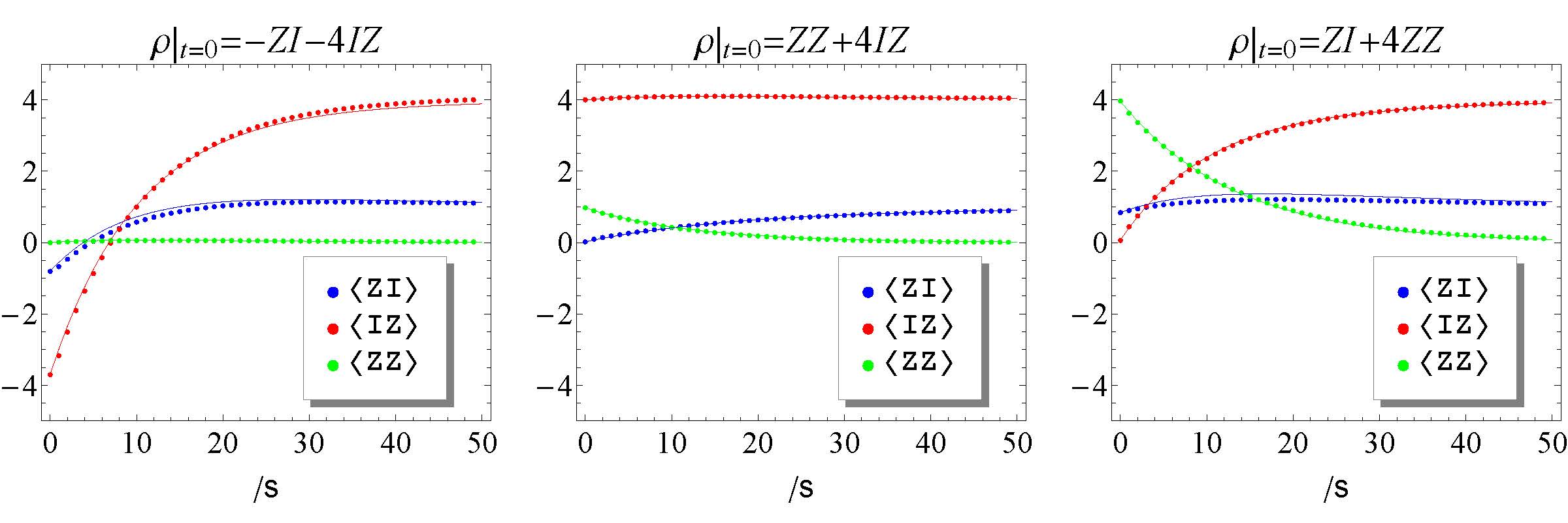}
\end{figure}
}
\end{itemize}
\newpage
\begin{itemize}
\item{$\left\{ r_7, r_8, r_9 \right \} \approx \left\{ 3.495, 6.536, 0.0100 \right \}$
\begin{figure}[h]
\centering
\includegraphics[width=0.8\linewidth]{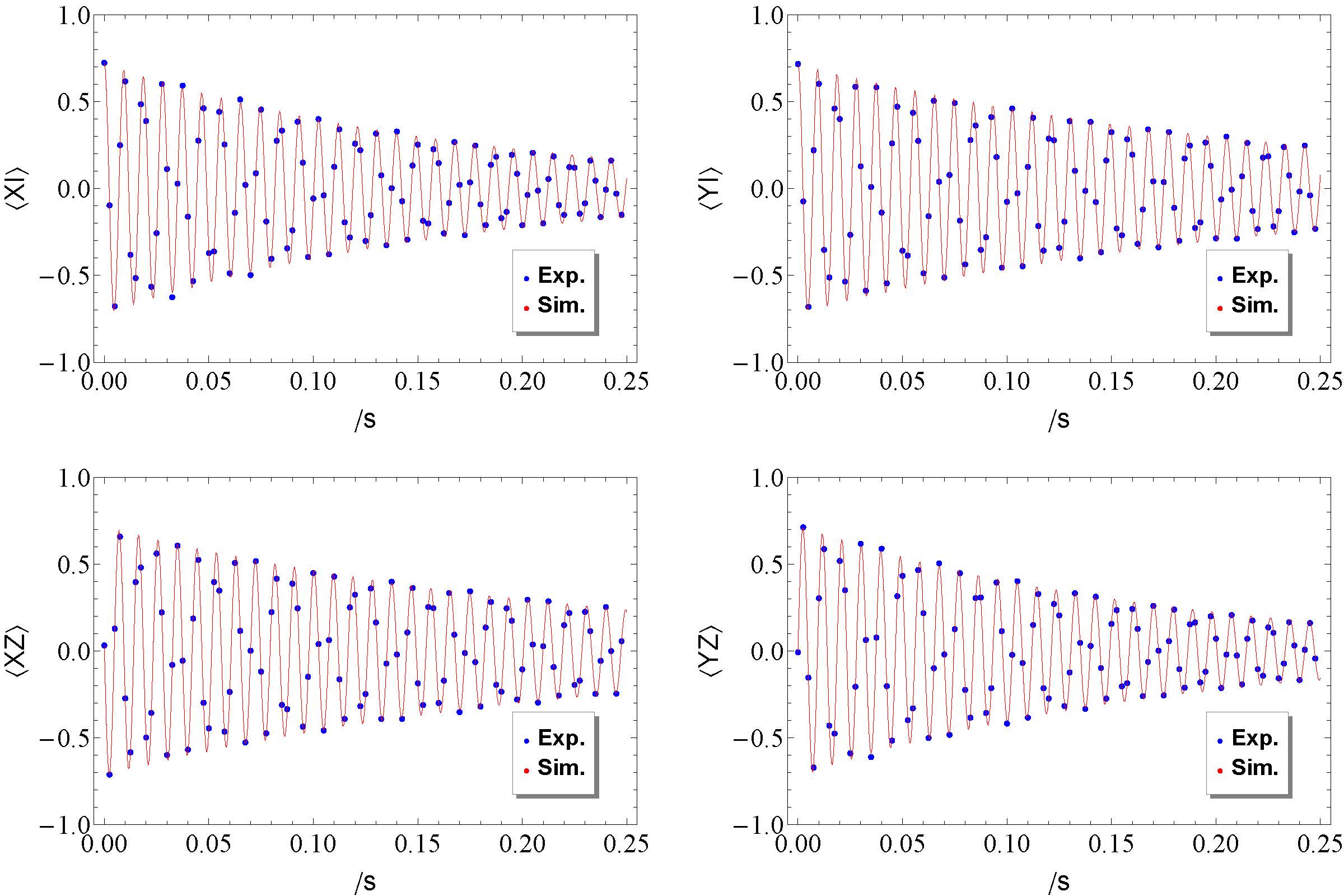}
\end{figure}
}
\end{itemize}
\begin{itemize}
\item{$\left\{ r_{10}, r_{11}, r_{12} \right \} \approx \left\{ 2.955, 6.118, 0.030 \right \}$
\begin{figure}[h]
\centering
\includegraphics[width=0.8\linewidth]{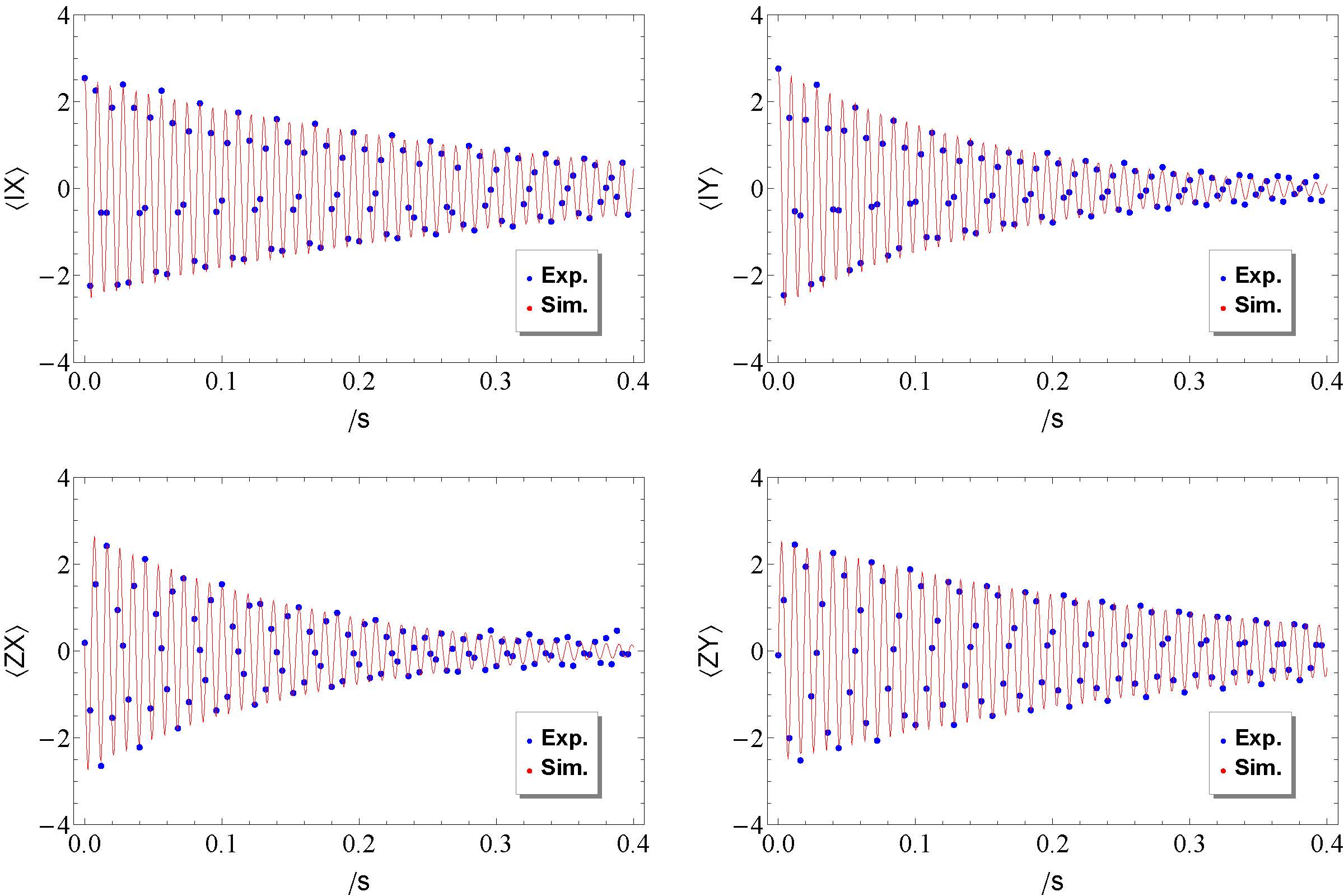}
\end{figure}
}
\end{itemize}
\begin{itemize}
\item{$\left\{ r_{13}, r_{14} \right \} \approx \left\{ 9.523, 0.008 \right \}$
\begin{figure}[h]
\centering
\includegraphics[width=0.75\linewidth]{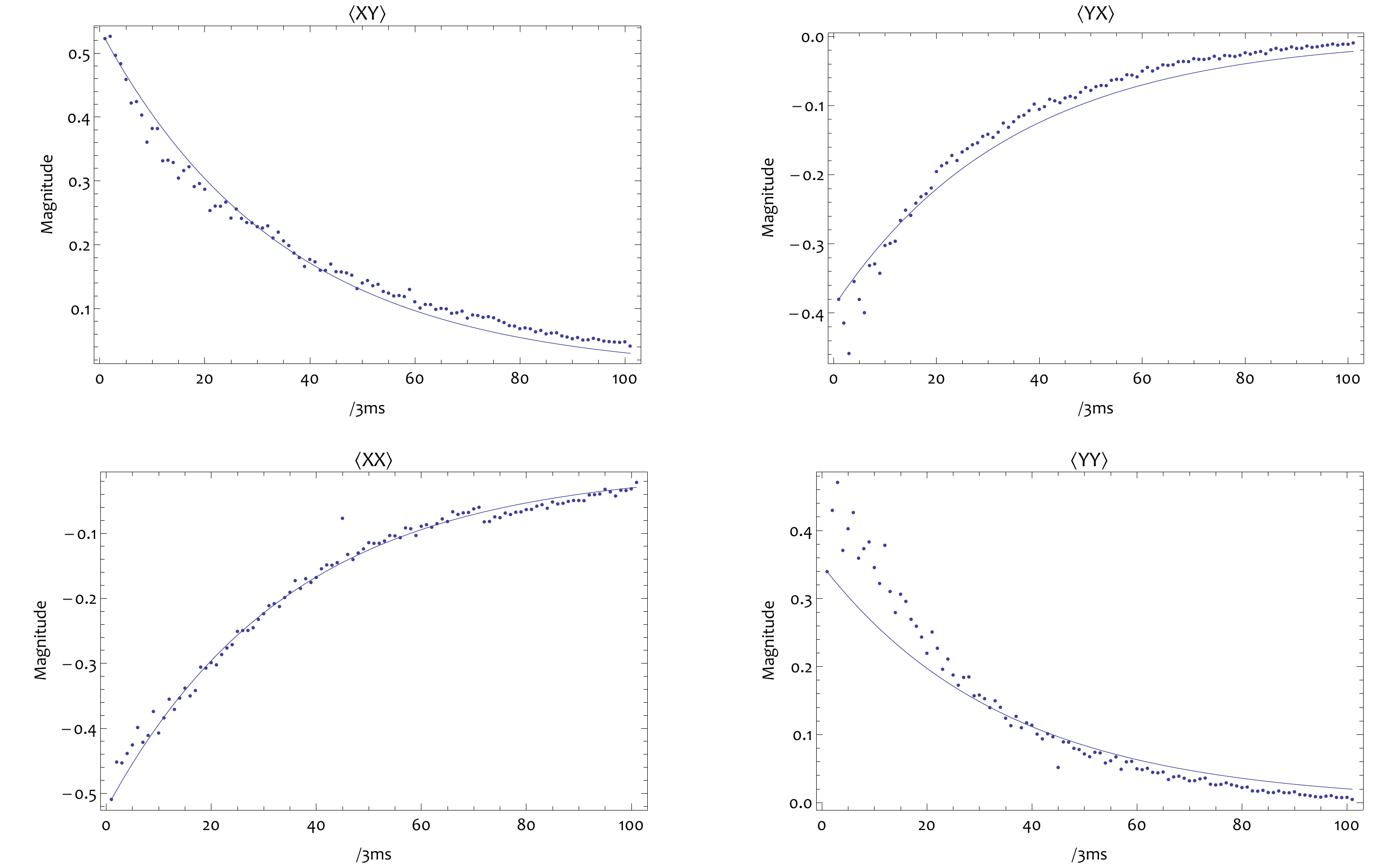}
\end{figure}
}
\end{itemize}

\newpage

\section{Robustness of Periodic Control Method for PPS Preparation}
Figure 2 shows the relative error of the prepared PPS due to imperfections of control fields present in the $^{13}$C channel ($\delta_{\text{C}} = \left| {B_{real}^{\text{C}} - B_{ideal}^{\text{C}}} \right|{\rm{/}}\left| {B_{ideal}^{\text{C}}} \right|$) and $^1$H channel ($\delta_{\text{H}} = \left| {B_{real}^{\text{H}} - B_{ideal}^{\text{H}}} \right|{\rm{/}}\left| {B_{ideal}^{\text{H}}} \right|$). The relative error is characterized by
\begin{equation}
\delta = \left\| {{\rho _{real}} - {\rho _{pps}}} \right\|/\left\| {{\rho _{pps}}} \right\|.
\end{equation}
It can be easily seen that the periodic control method is quite robust to the control imperfections.
\begin{figure}[h]
  \begin{minipage}[c]{0.6\textwidth}
    \includegraphics [width=0.8\textwidth]{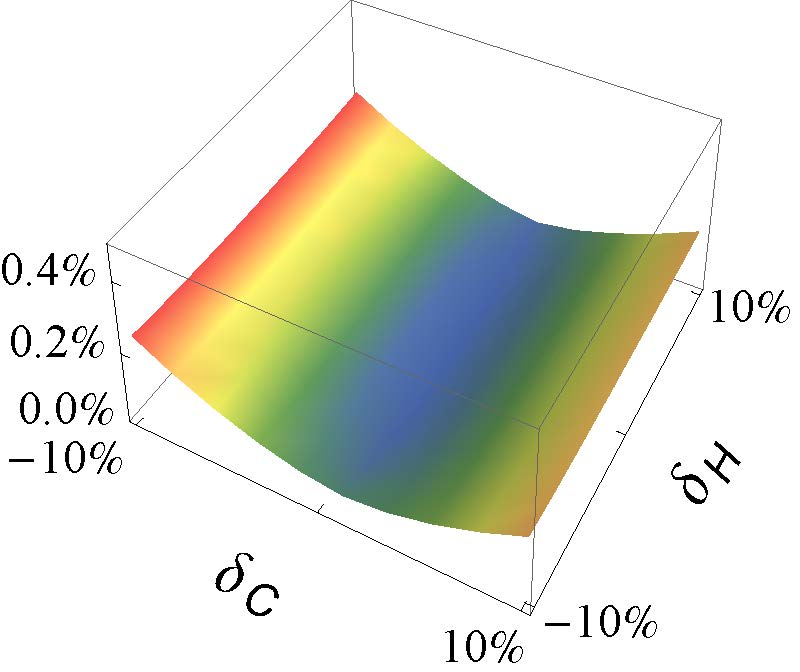}
  \end{minipage}\hfill
  \begin{minipage}[c]{0.4\textwidth}
    \caption{Simulation result: relative error of the prepared PPS due to imperfections of control fields present in the $^{13}$C channel and $^1 $H channel.}
    \label{Robustness}
  \end{minipage}
\end{figure}


\begin{thebibliography}{28}
\bibitem{DP} D. Dong and I. R. Petersen, IET Control Theory Appl. \textbf{4}, 2651 (2010).

\bibitem{KKSS} N. Khaneja \emph{et al}., J. Magn. Reson. \textbf{162}, 311 (2003); N. Khaneja \emph{et al}., J. Magn. Reson. \textbf{172}, 296 (2005); M. Lapert, Y. Zhang, M. Braun, S. J. Glaser, and D. Sugny, Phys. Rev. Lett. \textbf{104}, 083001 (2010).

\bibitem{RMBL} L. J. Schulman, Tal Mor, and Y. Weinstein, Phys. Rev. Lett. \textbf{94}, 120501 (2005); C. A. Ryan, O. Moussa, J. Baugh, and R. Laflamme, Phys. Rev. Lett. \textbf{100}, 140501 (2008).

\bibitem{SW} S. G. Schirmer and X. Wang, Phys. Rev. A \textbf{81}, 062306 (2010).

\bibitem{Entanglement} S. Diehl \emph{et al}, Nature Phys. \textbf{4}, 878 (2008); F. Verstraete, M. M. Wolf and J. I. Cirac, Nature Phys. \textbf{5}, 633 (2009); A. Pechen, Phys. Rev. A \textbf{84}, 042106 (2011); F. Ticozzi and L. Viola, Quantum Inform. Compu. \textbf{14}, 0265 (2014).

\bibitem{Lin} Y. Lin \emph{et al}., Nature (London) \textbf{504}, 415 (2013).

\bibitem{RTJS} F. Reiter, L. Tornberg, G. Johansson, and A. S. S{\o}rensen, Phys. Rev. A \textbf{88}, 032317 (2013).

\bibitem{SKVCG} M. J. A. Schuetz, E. M. Kessler, L. M. K. Vandersypen, J. I. Cirac, and G. Giedke, Phys. Rev. Lett. \textbf{111}, 246802 (2013).


\bibitem{D08} D. D'Alessandro, \emph{Introduction to Quantum Control and Dynamics} (Chapman $\&$ Hall, London, 2008).

\bibitem{A} C. Altafini, J. Math. Phys. \textbf{44}, 2357 (2003); C. Altafini, Phys. Rev. A \textbf{70}, 062321 (2004).


\bibitem{Y} H. Yuan, IEEE Trans. Autom. Control \textbf{55}, 955 (2010); H. Yuan, Syst. Control Lett. \textbf{61}, 1085 (2012).

\bibitem{RBR} P. Rooney, A. Bloch and C. Rangan, arXiv:1201.0399v1, (2012).

\bibitem{R} P. Rooney, Ph.D. thesis, University of Michigan, 2012.

%\bibitem{BM} F. Blanchini and S. Miani, \emph{Set-Theoretic Methods in Control} (Birkh\"{a}user Boston, 2008).

%\bibitem{A06} E. Asarin \emph{et al}., in \emph{Proc. of the 2006 IEEE Conference on Computer Aided Control Systems Design}, pp. 1582-1587, 2006.

\bibitem{M14}  O. Maler, EPTCS \textbf{140} 48 (2014).

\bibitem{KV} A. Kurzhanski and I. V\'{a}lyi, \emph{Ellipsoidal Calculus for Estimation and Control} (Birkh\"{a}user Basel, 1997).

\bibitem{ABDM} E. Asarin, O. Bournez, T. Dang and O. Maler, in \emph{Hybrid Systems: Computation and Control 2000}, edited by N. Lynch, B. H. Krogh (Springer, 2000), vol. 1790 of \emph{Lecture Notes in Computer Science}, pp. 20-31.

\bibitem{PY} P. M. Pardalos and V. Yatsenko, \emph{Optimization and Control of Bilinear Systems: Theory, Algorithms, and Applications}, (Springer, New York, 2008).


\bibitem{Lindblad} G. Lindblad, Commun. Math. Phys. \textbf{48}, 119 (1976).

\bibitem{BP} H.-P. Breuer and F. Petruccione, \emph{The Theory of Open Quantum Systems} (Oxford University Press, Oxford, 2002).

\bibitem{NC} M. A. Nielsen and I. L. Chuang, \emph{Quantum Computation and Quantum Information} (Cambridge University Press, Cambridge, England, 2010).


%\bibitem{RH} \'{A}. Rivas and S. F. Huelga, \emph{Open Quantum Systems: An Introduction} (Springer, New York, 2012).

\bibitem{K} I. Kurniawan, Ph.D. thesis, Universit{\"{a}}t W{\"{u}}rzburg, 2009.

\bibitem{S} See Supplemental Material for more information.

\bibitem{S90} O. W. Sprensen, \emph{J. Magn. Reson.} \textbf{86}, 435 (1990); J. Stoustrup \emph{et al}, Phys. Rev. Lett. \textbf{74}, 2921 (1995).

\bibitem{LSA} D. A. Lidar, A. Shabani, and R. Alicki, Chem. Phys. \textbf{82} 322 (2006).


\bibitem{FP}  C. A. Floudas and P. M. Pardalos (Eds.), \emph{Encyclopedia of Optimization} (Springer, New York, p3166, 2009).

\bibitem{L} M. H. Levitt, \emph{Spin Dynamics: Basics of Nuclear Magnetic Resonance} (John Wiley \& Sons Ltd, England, 2008).

\bibitem{CPH} D. G. Cory, M. D. Price, and T. F. Havel, Physica D \textbf{120}, 82 (1998).

\bibitem{GC} N. A. Gershenfeld and I. L. Chuang, Science \textbf{275}, 350 (1997).

\bibitem{KCL} E. Knill, I. Chuang, and R. Laflamme, Phys. Rev. A \textbf{57}, 3348 (1998).

\bibitem{PPS} E. Knill, R. Laflamme, R. Martinez and C.-H. Tseng, Nature (London) \textbf{404}, 368 (2000); U. Sakaguchi, H. Ozawa and T. Fukumi, Phys. Rev. A \textbf{61}, 042313 (2000); J. A. Jones, Prog. Nucl. Magn. Reson. Spectrosc. \textbf{38}, 325 (2001).

\bibitem{P} M. Pravia \emph{et al}., Concept. Magn. Reson. \textbf{11}, 225 (1999).

\bibitem{P01}  X. Peng \emph{et al}., Chem. Phys. Lett. \textbf{340}, 509 (2001).

\bibitem{XYS} F. Xue, S. X. Yu and C. P. Sun, Phy. Rev. A \textbf{73}, 013403 (2006).

\bibitem{L} S. Lloyd, Nature (London) \textbf{406}, 1047 (2000).

%\bibitem{PR} A. Pechen and H. Rabitz, Phys. Rev. A \textbf{73}, 062102 (2006); F. Shuang, A. Pechen, T. S. HO, and H. Rabitz, J. Chem. Phys. \textbf{126}, 134303 (2007); R. Romano and D. D'Alessandro, Phys. Rev. Lett. \textbf{97}, 080402 (2006).
\end{thebibliography}

\begin{thebibliography}{28}
%\bibitem{G10} Jean Gallier, \emph{Geometrical Methods and Applications: For Computer Science and Engineering} (Springer, 2010).

\bibitem{A} C. Altafini, J. Math. Phys. \textbf{44}, 2357 (2003); C. Altafini, Phys. Rev. A \textbf{70}, 062321 (2004).

\bibitem{K} I. Kurniawan, Ph.D. thesis, Universit{\"{a}}t W{\"{u}}rzburg, 2009.

\bibitem{SW} S. G. Schirmer and X. Wang, Phys. Rev. A \textbf{81}, 062306 (2010).

\bibitem{A03} C. Altafini, J. Math. Phys. 44, 2357 (2003).

\bibitem{NC} M. A. Nielsen and I. L. Chuang, \emph{Quantum Computation and Quantum Information} (Cambridge University Press, Cambridge, 2010).

\bibitem{Y} H. Yuan, IEEE Trans. Autom. Control \textbf{55}, 955 (2010).

\bibitem{Rooney12} P. Rooney, Ph.D. thesis, University of Michigan, 2012.

\bibitem{S86} Alexander Schrijver, \emph{Theory of Linear and Integer Programming} (John Wiley \& Sons, 1986).

\bibitem{HJ} R. A. Horn and C. R. Johnson, \emph{Matrix Analysis} (Cambridge University Press, Cambridge, England, 2013).

\bibitem{KM} J. Kowalewski and L. M{\"{a}}ler, \emph{Nuclear Spin Relaxation in Liquids: Theory, Experiments, and Applications} (Taylor \& Francis, New York, 2006).

\end{thebibliography}
\end{document}